\documentclass[11pt]{article}
\usepackage{fullpage}
\usepackage{float}
\usepackage{epsfig}
\usepackage{amsmath}
\usepackage{amssymb}
\usepackage{amstext}
\usepackage{amsmath}
\usepackage{xspace}
\usepackage{latexsym}
\usepackage{verbatim}
\usepackage{url}
\usepackage[ruled,vlined]{algorithm2e}

\usepackage{tikz}
\usetikzlibrary{arrows,positioning,automata}
\usepackage[T1]{fontenc}
\usepackage[utf8]{inputenc}
\usepackage{authblk}

\newtheorem{theorem}{Theorem}[section]

\newtheorem{claim}{Claim}
\newtheorem{lemma}[theorem]{Lemma}

\newtheorem{corollary}[theorem]{Corollary}

\newcommand{\qed}{\mbox{\ \ \ }\rule{6pt}{7pt} \bigskip}

\renewcommand{\comment}[1]{}
\newenvironment{proof}{\noindent{\em Proof:}}{\hfill\qed}

\begin{document}
\markboth{G. Askalidis et al.}{Socially Stable Matchings}

\title{Socially Stable Matchings}
\author[1]{Georgios Askalidis}
\author[2]{Nicole Immorlica}
\author[1]{Emmanouil Pountourakis}
\affil[1]{Northwestern University}
\affil[2]{Microsoft Research \& Northwestern University}

\renewcommand\Authands{ and }

\maketitle
\begin{abstract}
In two-sided matching markets, the agents are partitioned into two sets. Each agent wishes to be matched to an agent in the other set and has a strict preference over these potential matches.   A matching is stable if there are no blocking pairs, i.e., no pair of agents that prefer each other to their assigned matches. 
In this paper we study a variant of stable matching motivated by the fact that, in most centralized markets, many agents do not have direct communication with each other.  Hence even if some blocking pairs exist, the agents involved in those pairs may not be able to coordinate a deviation.  We model communication channels with a bipartite graph between the two sets of agents which we call the {\it social graph}, and we study {\it socially stable matchings}.  A matching is socially stable if there are no blocking pairs that are connected by an edge in the social graph.  
Socially stable matchings vary in size and so we look for a maximum socially stable matching.  We prove that this problem is {\sf NP}-hard and, assuming the unique games conjecture, hard to approximate within a factor of $\frac{3}{2}-\epsilon$, for any constant $\epsilon$. We complement the hardness results with a $\frac{3}{2}$-approximation algorithm.

\end{abstract}

\section{Introduction}

The main functionality of many markets is to match agents to one another.  A labor market matches workers to firms; a marriage market matches men to women; an education market matches students to schools; an online marketplace matches buyers to sellers, and so on.  In these markets, agents have preferences over potential matches based on the characteristics of the match -- the experience and skill of a worker, the salary and location of a firm, the personality and physique of a potential spouse, the test scores of a student, etc.  Standard economic theory suggests that the equilibria matchings in such markets will exhibit a particular stability property.  Stable matchings, and those that will be produced as the outcome of matching markets, are matchings in which no pair of unmatched agents {\it block} the matching, i.e., prefer each other to their assigned spouse.  Intuitively, the rationale behind this prediction is that if two unmatched agents form a blocking pair, then they will deviate from the matching and thereby destroy it.  The seminal results of Gale and Shapley \cite{gale1962college} prove that stable matchings always exist and can be computed efficiently by a centralized algorithm.  Later work of Roth \cite{roth1984evolution}, \cite{roth1990new} showed that, remarkably, some centralized matching markets (most notably the labor market for medical interns organized by the National Residency Matching Program) do indeed employ algorithms similar to the one proposed by Gale and Shapley.  Furthermore, centralized markets that do not output stable matchings are statistically less likely to persist than their counterparts that do.  Various theoretical studies ~\cite{roth1990random}, \cite{ackermann2011uncoordinated}, \cite{abeledo1995paths} suggest that decentralized markets may converge to stable outcomes as well.

However, in many markets, the existence of certain blocking pairs seems implausible.  In particular, for a pair of agents to successfully block a matching, they must be aware of each other's existence, able to communicate the potential deviation, and able to enact it.  Each of these assumptions places fairly strong restrictions on the markets.  For example, in sports drafts, matchings of players to teams is typically implemented by a centralized unstable mechanism known as serial dictatorship.  Yet deviations are rarely observed as strong regulations prevent blocking pairs from enacting the deviations.  Sorority rush at certain colleges also employs a centralized unstable mechanism.  Although some theory~\cite{mongell1991sorority} suggests that agents manipulate the mechanism into producing stable matches by reporting altered preferences, another explanation for the success of these mechanisms is that potential blocking pairs have difficulty communicating the deviation: for a girl to switch sororities, it is extremely useful for her to have a connection, a friend say, in the sorority.  
In more anonymous settings, like online marketplaces or online dating platforms, agents may not even be aware of each others' existence.  Rather than having agents rank each other explicitly (which requires agents to know one-another), these markets employ centralized mechanisms which infers agents' preferences through various means (questionnaires, cookies on users' machines, etc.).  The mechanism then selects a matching based on these preferences.  Now whether this matching is stable depends on whether the potential blocking pairs are aware of each other's existence.  

Similar stories can be told for decentralized matching markets.  In decentralized labor markets, for example, it is well-known that workers find out about new jobs more through personal contacts than any other means (see, for example, the seminal study of Granovetter, \cite{granovetter1973strength}).  Taken to the extreme, this suggests that, in a fixed matching, the only feasible blocking agents for a given firm are the friends of the workers employed by that firm.  Motivated by this observation, Arcaute and Vassilvitskii~\cite{arcaute2009social} study what they call {\it locally stable matchings}, or matchings stable with respect to this feasibility constraint on blocking pairs.  Their work and followup work showed that locally stable matchings exist and certain random dynamics converge to these matchings in polynomial time with only limited memory~\cite{hoefer2011local}.  Furthermore, they observed that, unlike stable matchings, locally stable matchings can have different sizes.  A natural design goal is to maximize the number of matched agents.  Accordingly, followup work~\cite{cheng2012maximum} proved that stable matchings are $\frac{3}{2}$-approximations to maximum-cardinality locally stable matchings for a very restrictive special case in which workers that are ``popular'' (matched in every stable matching) have no friends that are ``unpopular'' (single in some stable matching) that are also interested in some of the same firms.  They also conjectured that there is a $\frac{3}{2}$-approximation for the general case.  Other work~\cite{hoefer2012locally} proved $\frac{3}{2}$ is the best-possible approximation factor assuming the unique games conjecture.  

In this paper, we focus on centralized matching markets and restrict feasible blocking pairs to be a subset of all pairs.  This subset is represented by a bipartite graph called the {\it social graph} and is meant to capture agents with mutual knowledge of each other and direct channels of communication.  For example, a given girl will be connected to all sororities in which she has friends; a given buyer will be connected to all sellers with whom he has conducted prior business transactions; a man in an online dating platform will be connected with all the women he knows socially outside of the platform.  We note that, especially for online platforms, this social graph is something the mechanism itself can infer: agents that are in touch with each other are those that have messaged each other on the site or engaged in prior transactions.  Our goal is to compute a matching with no blocking pairs that form an edge in the social graph.  We call such a matching a {\it socially stable matching}.  Clearly any stable matching is also socially stable.   However the converse is not true.  Furthermore, as is the case for locally stable matchings, not all socially stable matchings have the same size.  Thus we focus on computing maximum socially stable matchings.  In this paper, we provide a centralized algorithm that computes a socially-stable matching which is at $\frac{2}{3}$ the size of the maximum socially-stable matching, yielding a $\frac{3}{2}$-approximation.  We also show our problem is {\sf NP}-hard and hard to approximate to within $\frac{3}{2}-\epsilon$ assuming the widely-believed unique games conjecture.

Socially stable matchings are clearly related to the locally stable matchings studied in prior work~\cite{arcaute2009social,cheng2012maximum,hoefer2011local,hoefer2012locally}, but there are some subtle and important differences.  From a modeling perspective, locally stable matchings ask us to believe that friends communicate opportunities to each other although they are direct competitors.  In particular, why should an agent inform his friends of his match if his friend might form a blocking pair with his match and leave him single?  In socially stable matchings, the incentives of the communicating agents (pairs in the social graph) are aligned.  This modeling change also causes a key technical difference: social stability is fundamentally a static notion.  The feasible blocking pairs are defined with respect to the input and not with respect to the proposed matching as is the case for locally stable matchings.  Due to this technical difference, neither model can be reduced to the other.  Furthermore, the static properties of our model enable us to suggest an intuitive algorithmic technique: prioritize matches that appear as edges in the social graph.  It is this intuition that allows us to get a tight approximation result in our model despite the illusiveness of a similar result in the local stability model.

\section{Preliminaries}

The agents of a standard matching market are partitioned into two sets: a set of men $M$ and a set of women $W$. Each man $m\in M$ has a strict preference order $\succ_m$ over $W\cup\{m\}$, and each woman $w\in W$ has a strict preference order $\succ_w$ over $M\cup \{w\}$.  We say that an agent $a$ prefers $b$ to $b'$ if $b\succ_a b'$, and that $a$ prefers being unmatched to being matched with $b$ if $a\succ_a b$.  The set of agents $\{b:a\succ_a b\}$ are said to be {\it unacceptable} to $a$; the remaining agents $\{b:b\succ_aa\}$ are said to be {\it acceptable}.

In addition to this standard input, we augment the model to take into account social ties between the men and women.  We respresent these social ties by a bipartite graph $G=(M\cup W,E)$ which we call the {\it social graph}.  We say $(m,w)$ are friends if $(m,w)$ is an edge in $G$, and interpret this to mean that agents $m$ and $w$ have a direct channel of communication.  Note that there is no correlation between the social graph and the preferences: even if $(m,w)$ is not an edge in $G$, it can be that $m$ and/or $w$ prefer each other to being single and/or to being matched to other candidates.

A matching $\mu$ is a function mapping $M\cup W$ to itself such that $\mu(m)=w$ if and only if $\mu(w)=m$.  If for some agent $a$, $\mu(a)=a$, then we say that agent $a$ is {\it single} in matching $\mu$.  The cardinality of a matching is the number of matched (i.e., not-single) agents.  When convenient, we sometimes indicate a matching by listing the matched pairs; the unlisted agents are understood to be single.

A matching $\mu$ is {\it individually rational} if, for all agents $a$ such that $\mu(a)\not=a$, $\mu(a)\succ_a a$, i.e., $a$ prefers his or her assigned match to being single.   A pair $(m,w)$ are a {\it blocking pair} for a matching $\mu$ if $w\succ_m\mu(m)$ and $m\succ_w\mu(w)$, i.e., if they both strictly prefer each other to their assigned matches.  A matching is {\it stable} if it is individually rational and has no blocking pairs.  

It is known that stable matchings always exist; all stable matchings have the same cardinality; and that the stable matchings form a lattice structure.  In particular, there is a single stable matching $\mu_M$, called the {\it man-optimal} stable matching, that is preferred by all men to all other stable matchings $\mu$, i.e., $\forall m\in M$, either  $\mu_M(m)=\mu(m)$ or $\mu_M(m) \succ_m\mu(m)$.  Furthermore, this matching can be computed efficiently using the {\it man-proposing deferred-acceptance algorithm}~\cite{gale1962college}.   This algorithm proceeds as follows: the tentative matching $\mu$ is initialized to be the empty matching, i.e., for all agents $a$, $\mu(a)=a$.  In each round, a single man proposes to his most-preferred woman who has not yet rejected him.  When a woman $w$ receives a proposal from a man $m'$, she tentatively accepts if she prefers $m'$ to her current match $\mu(w)$ and rejects her current match, if any.  A more complete description of this algorithm can be found in~\cite{gale1962college}.  We use two key properties of this algorithm: 1) a woman is single at the end of the process if and only if she received no proposals, and 2) a man is single at the end of the process if and only if he has proposed to every acceptable woman on his preference list.

We are particularly interested in matchings which are stable with respect to the social ties.  We say a pair $(m,w)$ are a {\it social blocking pair} for a matching $\mu$ if $(m,w)$ are a blocking pair for $\mu$ and $(m,w)$ is an edge in the social graph $G$.  Similar to stable matchings, we define a matching $\mu$ to be {\it socially stable} if it is individually rational and has no social blocking pairs.  

This definition captures the idea that a blocking pair only constitutes a threat to the stability of a matching if the two agents have direct communication and are able to coordinate a deviation. It's straightforward to see that if $G$ is the empty graph then no pair can be socially-blocking and hence all matchings are socially stable. In contrast, if $G$ is the complete bipartite graph then the notion of social stability coincides with the traditional notion of stability.  In general, the set of socially stable matchings is a superset of the set of stable matchings.

In contrast with the case of stable matchings, not all socially-stable matchings have the same size. For example, given an empty graph, any matching is socially stable.   Hence, a natural optimization problem arises: Given a set of preferences $\{\succ_m \mid m\in M\}$ and $\{\succ_w \mid w\in W\}$, and a social graph $G=(M\cup W, E)$, find a maximum-cardinality socially stable matching. We call this the \textsc{Soc-Stable} problem.  As we will see, this problem is {\sf NP}-hard, and so we develop approximation algorithms for it.  An algorithm $A$ is an $\alpha$-approximation for the \textsc{Soc-Stable} problem if, on every instance, it outputs a matching $\mu_A$ whose cardinality $|\mu_A|$ is at least $\frac{1}{\alpha}|\hat\mu|$ for the maximum socially stable matching $\hat\mu$.

\paragraph{Example}
Consider an instance $I=(P, G=(M\cup W, E))$ with two men, $m_1$ and $m_2,$ and two women, $w_1$ and $w_2$ with preferences as shown in Figure 1.
\begin{figure}[h!]\label{example}
\setlength{\unitlength}{1.2cm}
\centering

    \begin{center}
\begin{tikzpicture}
\tikzstyle{every node}=[draw,circle,fill=black,minimum size=3pt,
                            inner sep=0.3pt]
                            
\draw (0,2) node (m1) [label=above left:$m_1$]{};
\draw (2,2) node (w1) [label=above right:$w_1$]{};
\path (m1) edge (w1);

\draw (0,1) node (m2) [label=below left:$m_2$]{};
\draw (2,1) node (w2) [label=below right:$w_2$]{};
\path (m2) edge (w2);

\tikzset{every node/.style={fill=white}} 

\node[anchor=north west,text width=3cm] (note2) at (-5,2.3)
   {$\displaystyle \succ_{m_1}: w_2, w_1, m_1$\\
$\displaystyle \succ_{m_2}: w_2, w_1, m_2$\\
$\displaystyle \succ_{w_1}: m_1, w_1$\\
$\displaystyle \succ_{w_2}: m_1, m_2, w_2$
}; 
\end{tikzpicture}
\end{center}

%\begin{picture}(1.5,1.5)
%\put(-2,0.9){\circle*{0.05}} %m_1
%\put(-2,1.1){$\displaystyle m_{1}$}

%\put(-2,0.2){\circle*{0.05}}%m_2
%\put(-2,0.0){$\displaystyle m_{2}$}

%\put(-0.5,0.9){\circle*{0.05}} %w_1
%\put(-0.5,1.1){$\displaystyle w_{1}$}

%\put(-0.5,0.2){\circle*{0.05}}%w_2
%\put(-0.5,0.0){$\displaystyle w_{2}$}

%\put(-2,0.9){\line(1,0){1.5}}

%\put(-2,0.2){\line(1,0){1.5}}

%\put(1.5, 1.4){$\displaystyle P$}
%\put(0.8,1.1){$\displaystyle \succ_{m_1}: w_2, w_1, m_1$}
%\put(0.8,0.7){$\displaystyle \succ_{m_2}: w_2,m_2, w_1$}
%\put(0.8,0.35){$\displaystyle \succ_{w_1}: m_1, w_1, m_2$}
%\put(0.8,0.0){$\displaystyle \succ_{w_2}: m_1, m_2, w_2$}
%\end{picture}
\caption{Example of an instance for \textsc{Soc-Stable}}
\end{figure}
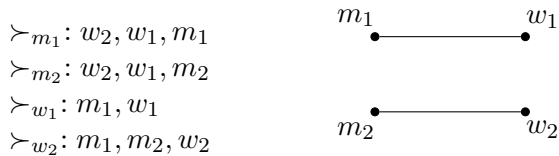
Let $G$ be such that $E(G)=\{ (m_1,w_1), (m_2,w_2)\}$. Then the only stable matching $\mu$ for this instance is $\mu=\{(m_1,w_2)\}$. As $\mu$ is stable, it is also socially stable, and so there is a socially stable matching of cardinality one. However, there is a second larger socially stable matching $\hat{\mu}$ of cardinality two which matches $\hat\mu=\{(m_1,w_1), (m_2,w_2)\}$. Indeed $m_1$ and $w_2$ prefer each other to their assigned matches under $\hat{\mu}$ and so $\hat\mu$ is not stable, but since $(m_1,w_2)\not\in E(G)$ they are not a social blocking pair and so $\hat\mu$ is socially stable.

\section{Approximation Results}
\label{sec:approx}

As shown in  Section~\ref{sec:hardness}, the \textsc{Soc-Stable} problem is {\sf NP}-hard and $\frac{3}{2}$-inapproximable assuming the widely-believed unique games conjecture.  Furthermore, there is a trivial $2$-approximation for the problem: simply compute a stable matching and return that.  Since a stable matching must be {\it maximal} with respect to the restrictions imposed by individual rationality, it is a $2$-approximation to the {\it maximum} individually rational matching which, in turn, is an upper-bound on the cardinality of the maximum socially stable matching.  

In this section, we develop a polynomial time algorithm that is a $\frac{3}{2}$-approximation.  The main idea behind the algorithm is to alter each woman's preference and give low priority to all the men that she is not connected with in the social graph, no matter their position in her preference list. We then use an idea introduced by Kir\'aly~\cite{kiraly2011better}: we first compute the man-optimal stable matching.  We then give a second chance to men that are single in this matching by promoting them in the women's preferences and allowing them to re-initiate their proposal process.

To formally define the algorithm, consider an instance of the \textsc{Soc-Stable} problem with men $M$, women $W$, preferences $\{\succ_m|m\in M\}$ and $\{\succ_w|w\in W\}$, and social graph $G = (M\cup W, E)$.  For each woman $w\in W$, we let $T_w$ denote the set of men that woman $w$ has promoted to the top of her preference list.  Throughout the algorithm, we maintain the invariant that $T_w$ is a prefix of $w$'s preference list.  We initiate $T_w=\{m:(m,w)\in E(G)\}$ and create a corresponding preference list $\succ'_w$ by promoting all men in $T_w$ to the top of $w$'s preference, preserving the relative order.  More formally, for all acceptable men $m$ and $m'$, we define $m\succ'_wm'$ if and only if one of the following three conditions holds: 1) $m\succ_wm'$ and $(m,w)\in E(G)$ and $(m',w)\in E(G)$, 2) $m\succ_wm'$ and $(m,w)\not\in E(G)$ and $(m',w)\not\in E(G)$, or 3) $(m,w)\in E(G)$ and $(m',w)\not\in E(G)$ (the set of unacceptable men and their relative ordering remains unchanged).  Note that, due to the third condition, it might be the case that $m\succ'_wm'$ while $m'\succ_wm$.  

Our algorithm first computes a man-optimal stable matching $\mu_M$ for a market with preferences $\{\succ_m\}$ for the men and $\{\succ'_w\}$ for the women.  We initialize our tentative socially stable matching $\mu$ to be this man-optimal stable matching $\mu_M$.  We then give every single man a second chance.  To this end, we introduce a variable  $s_m\in\{0,1\}$ that tracks whether a man has received a second chance and initialize $s_m:=0$ for all $m$.  We repeatedly select a currently single man $m^*$ with $s_{m^*}:=0$, set $s_{m^*}:=1$ and $T_w:=T_w\cup\{m^*\}$ for all women $w$, and promote $m^*$ in the women's preferences by redefining $\succ'_w$ to be $m\succ'_wm^*$ if and only if $m\succ_wm^*$ and $m\in T_w$ (the preference order $\succ'_w$ is unchanged on all other comparisons).  We then set the tentative matching $\mu$ to be the new man-optimal stable matching with respect to these altered preferences.  When the algorithm terminates, we return $\mu$ as the socially stable matching.  A pseudo-code description of the algorithm appears above.

\begin{algorithm}[t]
1. $\forall\ w\in W$, set $T_w=\{m:(m,w)\in E(G)\}$.\\
2. $\forall\ w\in W$ and $\forall\ m,m'\in M : m\succ_ww\wedge m'\succ_ww$ define $m\succ'_wm'$ if and only if\\
\ \ \ \ \ a. $m\succ_wm'$ and $(m,w)\in E(G)$ and $(m',w)\in E(G)$,\\
\ \ \ \ \ b. $m\succ_wm'$ and $(m,w)\not\in E(G)$ and $(m',w)\not\in E(G)$,\\
\ \ \ \ \ c. or $(m,w)\in E(G)$ and $(m',w)\not\in E(G)$.\\
For all remaining pairs, let $\succ'=\succ$.\\
3. For every man $m$, set $s_m=0$.\\
4. Let $\mu$ be the outcome of the man-proposing deferred acceptance algorithm.\\
\While{$\exists\ m^*:\mu(m^*)=m^*\wedge s_{m^*}=0$}
{
	 1. Set $s_{m^*}:=1$ and, $\forall\ w\in W$, set $T_w:=T_w\cup\{m^*\}$.\\
	 2. $\forall\ w\in W$ and $\forall\ m\in T_w$ set $m^*\succ'_wm$ if and only if $m^*\succ_wm$.\\
	 3. $\forall\ w\in W$ and $\forall\ m\not\in T_w$ set $m^*\succ'_wm$.\\
	 4. Let $\mu$ be the outcome of the man-proposing deferred acceptance algorithm.\\
}
\caption{socGS}
\label{socgs}
\end{algorithm}

Note that this algorithm terminates in polynomial time as we run the polynomial-time man-proposing deferred-acceptance algorithm at most $|M|+1$ times.  We would like to prove two additional properties of this algorithm.  First, we show that it computes a socially stable matching.  Second, we show the socially stable matching is a $\frac{3}{2}$-approximation to the maximum socially stable matching.

\begin{theorem}\label{strictsoc}
Algorithm socGS computes a socially stable matching.
\end{theorem}

\begin{proof}
Let $\mu$ be the matching output by the algorithm and note that, by the properties of the deferred-acceptance algorithm, $\mu$ is stable with respect to the men's preferences $\succ_m$ and the women's final altered preferences $\succ'_w$.  Suppose that $\mu$ is not socially stable with respect to the true preferences $\{\succ_m,\succ_w\}$ and let $(m,w)$ be a socially blocking pair.  Then $w\succ_m\mu(m)$, $m\succ_w\mu(w)$, and $(m,w)\in E(G)$.   

We consider two cases based on whether $(\mu(w),w)\in E(G)$.  First suppose $(\mu(w),w)\in E(G)$.  Then at any point in the algorithm, $\succ_w$ and $\succ'_w$ agree on the ordering of $m$ and $\mu(w)$.  Thus $w\succ_m\mu(m)$ and $m\succ'_w\mu(w)$, contradicting the stability of $\mu$.

Next suppose $(\mu(w),w)\not\in E(G)$.  Then the algorithm initializes $m\succ'_w\mu(w)$.  The algorithm may later promote $\mu(w)$, but only to his position in the original preferences, and so even after promotion $m\succ'_w\mu(w)$.  This again contradicts the stability of $\mu$.
\end{proof}

We next argue that the algorithm is a $\frac{3}{2}$-approximation.  The proof follows a standard technique also used by Iwama et al. ~\cite{iwama20071} in the context of stable matchings: to prove a matching $\mu$ is a $\frac{3}{2}$-approximation to a matching $\hat\mu$, consider the multi-graph induced by the union $\mu\cup \hat{\mu}$ and note that every component consists of cycles or \textit{alternating paths}, i.e. paths whose edges belong to $\mu$ and $\hat{\mu}$ alternately. To prove that $|\mu|\geq\frac{2}{3}\hat\mu$, it is enough to show that there is no alternating path of size three that starts and ends with an edge of $\hat{\mu}$ (as in all other possible components the inequality holds).

\begin{theorem} Let $\mu$ be the outcome of Algorithm socGS and $\hat{\mu}$ be the maximum socially stable matching. Then $|\mu|\geq \frac{2}{3}|\hat{\mu}|$.
\end{theorem}

\begin{proof}
Suppose that there exists an alternating path $(w^*,m,w,m^*)$ in the multi-graph induced by the union $\mu\cup\hat{\mu}$ such that $(m,w^*)\in\hat\mu$, $(m,w)\in\mu$, and $(m^*,w)\in\hat\mu$. We will contradict the social stability of $\hat\mu$.  

Note both $w^*$ and $m^*$ are single in $\mu$.
Since $w^*$ is single at the end of the algorithm, we show that  she never received an offer in {\it any} iteration of the man-proposing deferred-acceptance algorithm. We argue that every woman's match can only improve in subsequent iterations. This is because  the outcome of man-proposing deferred acceptance algorithm, without including the man that was given as second chance at this iteration, will be exactly the outcome of the previous iteration since he ended up single. Therefore, since it is known that the order of proposals does not matter, we can assume that we first produce the matching of the previous iteration by letting these men propose and then we allow the man that was given a second chance to start proposing. It is known also that the women are weakly improved at any step of the man-proposing deferred acceptance algorithm, which implies that they are at least as happy as the previous iteration. Hence, $w^*$ was single at every iteration which only possible if she received no offer in any iteration.  This implies that  $m$ never proposed to $w^*$ even though she was acceptable to him and he to her (by the individual rationality of $\hat\mu$.  This means that 
\begin{equation}\label{eq2}m\mathrm{\ was\ never\ single,}\end{equation}
and 
\begin{equation}\label{eq1}w\succ_mw^*.\end{equation}
Furthermore, since $m^*$ is single at the end of the algorithm, at some point $m^*$ was given a second chance (i.e., selected as the single man in the while-loop of the algorithm).  At this point, he was promoted in woman $w$'s list, proposed to her, and was still rejected during that iteration of the man-proposing deferred-acceptance algorithm.  As women's tentative matches only improve in this algorithm, this means that $m\succ'_wm^*$.  Thus since $m^*$ was added to $T_w$ in this iteration and $T_w$ is a prefix of $w$'s list,
\begin{equation}\label{eq3}m\in T_w,\end{equation}
and 
\begin{equation}\label{eq4}m\succ_wm^*.\end{equation}
Finally, we observe that, since $m$ was never single and hence never chosen in the while-loop of the algorithm, and yet $m \in T_w$, it must be that $m$ was added to $T_w$ in the initialization proving $(m,w)\in E(G)$.  We also argued that $w\succ_mw^*$ and $m\succ_wm^*$, and so $(m,w)$ is a social blocking pair for $\hat\mu$, contradicting the social stability of $\hat{\mu}$.
\end{proof}

\paragraph{Example} The example below shows that this upper bound of the performance of Algorithm socGS is tight.
 
\begin{center}
\begin{tikzpicture}
\tikzstyle{every node}=[draw,circle,fill=black,minimum size=3pt,
                            inner sep=0.3pt]
\draw (0,3) node (m1) [label=above left:$m_1$]{};
\draw (2,3) node (w1) [label=above right:$w_1$]{};

\draw (0,2) node (m2) [label=above left:$m_2$]{};
\draw (2,2) node (w2) [label=above right:$w_2$]{};

\path (m1) edge (w1);
\draw (0,1) node (m3) [label=below left:$m_3$]{};
\draw (2,1) node (w3) [label=below right:$w_3$]{};
\path (m3) edge (w3);
\path (m2) edge (w3);
\tikzset{every node/.style={fill=white}} 

\node[anchor=north west,text width=4cm] (note2) at (-5,3.3)
   {$\displaystyle \succ_{m_1}: w_1, w_2, m_1$\\
$\displaystyle \succ_{m_2}: w_1, w_3, w_2,m_2$\\
$\displaystyle \succ_{m_3}: w_3,m_3$\\ \ \\
$\displaystyle \succ_{w_1}: m_2,m_1,w_1$\\
$\displaystyle \succ_{w_2}: m_1, w_2$\\
$\displaystyle \succ_{w_3}: m_2, m_3, w_3$
}; 
\end{tikzpicture}
\end{center}

Algorithm socGS will initiate preferences of the women as follows:
$$\succ'_{w_1}: m_1, m_2, w_1, \:\:  \succ'_{w_2}: m_1,w_1,\:\text{ and }\:\succ'_{w_3}: m_2,m_3,w_1$$
Then $m_1$ will propose to $w_1$ and they will be matched. $m_2$ will propose to $w_1$ and will be rejected since $m_1\succ'_{w_1} m_2$, so $m_2$ will propose to $w_3$ and they will be matched. After that, $m_3$ will propose to $w_3$, will be rejected and hence remain single at the end of the first run of the men-proposing Gale-Shapley algorithm. So algorithm socGS will promote him in the lists of all women that find him acceptable, set $s(m_3)=1$ and reactivate him. $m_3$ will propose to $w_3$ but she will reject him again since she prefers $m_2$ even in the new preferences. So the matching that socGS will return will be $\mu=\{(m_1,w_1), (m_2,w_3) \}$ which has size two. It's easy to see that $\hat{\mu}=\{(m_1, w_2), (m_2, w_1), (m_3, w_3)\}$ is the maximum socially stable matching and has size 3.

\section{{\sf NP}-Hardness and Hardness of Approximation}
\label{sec:hardness}

\par Having already shown that a $\frac{3}{2}$-approximation ratio is possible within polynomial time, we prove here a matching lower bound. We show that assuming the unique games conjecture, the \textsc{Soc-Stable} problem cannot be approximated in polynomial time within a ratio of $\frac{3}{2}-\epsilon$, for any constant $\epsilon$. We do that by showing a reduction from the notoriously {\sf NP}-hard problem of Independent Set. The reduction will give us as an obvious consequence that \textsc{Soc-Stable} is {\sf NP}-hard.

\par The Independent Set problem takes as an instance a pair $(G,k)$ of a graph and an integer $k\geq 0$ and asks if the graph $G$ contains a set $S\subseteq V(G)$ such that $|S|=k$ and for every $v_1, v_2\in S, (v_1,v_2)\notin E(G)$, i.e. no two nodes of $S$ are connected with an edge. It was one of the first problems that was proven to be {\sf NP}-hard.

\par Our reduction follows very closely the reduction that Hoefer and Wagner~\cite{hoefer2012locally} from Independent Set to the problem of calculating the maximum locally stable matching. Although the basic idea of the construction is the same, we need to adjust the details as well as the proofs to our problem.

\par Given a graph $G=(V, E)$ we create an instance $I=(\{\succ_m\},\{\succ_w\},G'=(M\cup W, E))$ of the \textsc{Soc-Stable} matching problem as follows: First we enumerate (arbitrarily but consistently) all the nodes of $G: v^1, v^2,\ldots, v^n$, where $n=|V(G)|$. For every $v^i\in V(G)$ we create four vertices of $G'$: $m_{v_1^i}, m_{v_2^i}\in M$ and $w_{v_1^i}, w_{v_2^i}\in W$. For every $v^i\in V(G)$ we denote by $N_G(v^i)=\{v^j \mid (v^i,v^j)\in E(G)\}$ the neighborhood of $v^i$ in $G$, i.e. the nodes that $v^i$ is connected with in $G$. Similarly, for a set $S\subset V(G)$ we denote by $N_G(S)=\bigcup_{v^i\in S}N_G(v^i)$ the union of the neighborhoods of the nodes in $S$.

\begin{figure}[h!]\label{example}
\setlength{\unitlength}{1.2cm}
    \begin{center}
\begin{tikzpicture}[>=stealth',shorten >=1pt,node distance=2cm,on grid,initial/.style    ={}]
\tikzstyle{every node}=[draw,circle,fill=black,minimum size=3pt,
                            inner sep=0.3pt]

%\draw (-3,2) node (v1) [label=above left:$v^1$]{};
%\draw (-4,2) node (v2) [label= below right:$v^2$]{};

\draw (0,2) node (m11) [label=above left:$m_{v^1_1}$]{};
\draw (2,2) node (w11) [label=above right:$w_{v^1_1}$]{};
%\path (m11) edge (w11);

\draw (0,1) node (m12) [label=below left:$m_{v^1_2}$]{};
\draw (2,1) node (w12) [label=below right:$w_{v^1_2}$]{};
%\path (m12) edge (w12);

\draw (4,2) node (m21) [label=above left:$m_{v^2_1}$]{};
\draw (6,2) node (w21) [label=above right:$w_{v^2_1}$]{};
%\path (m21) edge (w21);

\draw (4,1) node (m22) [label=below left:$m_{v^2_2}$]{};
\draw (6,1) node (w22) [label=below right:$w_{v^2_2}$]{};
%\path (m22) edge (w22);

    \tikzset{every node/.style={fill=white}} 
%    \path (m11)  edge [-,]   (w11);
%    \path (m12)  edge [-,]   (w12);

%\path (m21)  edge [-,]   (w21);
   % \path (m22)  edge [-,]   (w22);
  \path (m11)  edge [-,]   (w22);
  \path (m21)  edge [-,]   (w12);
%\node[anchor=north west,text width=3cm] (note2) at (-6,0.5)
  % {$\displaystyle \succ_{m_1}: w_2, w_1, m_1$\\
%$\displaystyle \succ_{m_1}: w_2, w_1, m_1$\\
%$\displaystyle \succ_{w_1}: m_1, w_1, m_2$\\
%$\displaystyle \succ_{w_2}: m_1, m_2, w_2$
%}; 
\end{tikzpicture}
\end{center}
\caption{Reduction from Independent Set: The gadget of some edge ($v^1,v^2)$}
\end{figure}
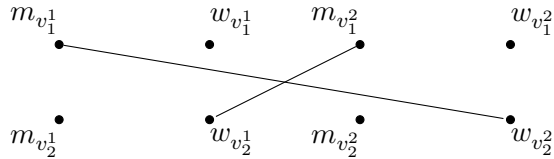

For every $v^i\in V(G)$ and every $v^j\in N_G(v^i)$ we add the edge $\{m_{v_1^i}, w_{v_2^j}\}$ to $E(G')$. See Figure 2 for an example of the transformation for an edge $(v^1,v^2)$. The preferences of the agents are as follows:\\

\begin{quote}
$\succ_{m_{v_1^i}}: w_{v_2^i}, \{w_{v_2^j}$ in increasing order on $j \mid v^j\in N_G(v^i)\}, w_{v_1^i}, m_{v_1^i}$\\
$\succ_{w_{v_2^i}}: m_{v_1^i}, \{m_{v_1^j}$ in increasing order on $j \mid v^j\in N_G(v^i)\}, m_{v_2^i}, w_{v_2^i}$\\
$\succ_{m_{v_2^i}}: w_{v_2^i}, m_{v_2^i}$\\
$\succ_{w_{v_1^i}}: m_{v_1^i}, w_{v_1^i}$\\
\end{quote}

\noindent
Notice that the enumeration of the vertices of $V(G)$ that we did in the beginning of our reduction shows up here, in the preferences of $m_{v_1^i}$ and $w_{v_2^i}$.   We prove the following two lemmas for the above construction.

\begin{lemma}\label{istosocstable}
If $G$ has an independent set $S$ of size $r$ then $I$ has a socially stable matching $\mu$ of size $n+r$.
\end{lemma}
\begin{proof}
We take $\mu=\{(m_{v_1^i}, w_{v_2^i})\mid v^i\in V(G)\setminus S\}\cup \{(m_{v_1^i}, w_{v_1^i}), (m_{v_2^i}, w_{v_2^i})\mid v^i\in S\}$. We obviously have that $|\mu|=n-|S|+2\cdot |S|=n+r$.
The pairs $(m_{v_1^i}, w_{v_2^i})$ are always socially stable. For the rest of the pairs, the independent set property guarantees that for every $v^j\in N_G(S)$, the matching of $(m_{v_1^j}, w_{v_2^j})$ keeps $(m_{v_1^i}, w_{v_1^i})$ and $(m_{v_2^i}, w_{v_2^i})$ socially stable. 
\end{proof}

The reverse direction is a little more technical.
\begin{lemma}\label{socstabletois}
If $I$ has a socially stable matching of size $n+r$ then $G$ has an independent set of size $r$.
\end{lemma}
\begin{proof}
First of all we can assume that $\mu$ is such that all $m_{v_1^i}$ are matched. If there is an $m_{v_1^i}$ that is not matched we can just take $w_{v_2^i}$ and match her to him. If some other $m_{v_1^j}$ gets single by this action we just do it repeatedly until no such $m_{v_1^i}$ exists. This will not break social stability of the matching nor decrease the size. 
\begin{claim}
If there exists $i\neq j$ such that $(m_{v_1^i}, w_{v_2^j})\in \mu$, then $(m_{v_1^j}, w_{v_2^i})\in \mu$ as well.
\end{claim}
\begin{proof}
Notice first, that by our construction and the fact that $(m_{v_1^i}, w_{v_2^j})\in \mu$ we get that $v^j\in N_G(v^i)$. And hence $(m_{v_1^j}, w_{v_2^i})\in E(G')$.  Suppose now that the statement is not true and take the minimum $i$ such that $(m_{v_1^i}, w_{v_2^j})\in \mu$ for some $j\neq i$. We first have that $i<j$. For suppose that $j<i$, then if $m_{v_1^j}$ is matched to some $w_{v_2^k}$, for some $k\neq j$, that contradicts the minimal choice of $i$. Hence it must be that $m_{v_1^j}$ is matched to $w_{v_1^j}$. Since $(m_{v_1^j}, w_{v_2^i})\in E(G')$, and $m_{v_1^j}$ prefers $w_{v_2^i}$ over $w_{v_1^j}$,  in order for $\mu$ to be socially stable it must be that $w_{v_2^i}$ is matched to someone that she prefers to $m_{v_1^j}$ and that can only be some $m_{v_1^k}$ for some $k<j<i$, which contradicts the minimal choice of $i$. So indeed, it must be  that $i<j$.\\
Since $(m_{v_1^j},w_{v_2^i})\notin \mu$ and they are connected with an edge in $G'$, in order for $\mu$ to be socially stable it needs to be that at least one of $m_{v_1^j}$ and $w_{v_2^i}$ gets a match under $\mu$ that he/she prefers to the other.

Suppose that $w_{v_2^i}$ is matched to some $m_{v_1^k}$ for some $k<j$ (these are the only men that $w_{v_2^i}$ prefers from $m_{v_1^j}$, since her top choice, $m_{v_1^i}$, is matched with someone else). Then $m_{v_1^i}$ prefers $w_{v_2^k}$ to $w_{v_2^j}$ (since $k<j$) and $w_{v_2^k}$ prefers $m_{v_1^i}$ to her current match unless her current match is a $m_{v_1^m}$ for some $m\neq k$ and $m<i$. But that would contradict the minimal selection of $i$. So it can't be the case that $w_{v_2^i}$ is matched under $\mu$ to someone she prefers to $m_{v_1^j}$.

Suppose now that $m_{v_1^j}$ is matched to some $w_{v_2^k}$ for some $k<i$. That would mean that $m_{v_1^k}$ is not matched to $w_{v_2^k}$ and if he is matched to some $w_{v_2^m}$ for some $m\neq k$, that would contradict the minimal choice of $i$. Hence $m_{v_1^k}$ must be matched to $w_{v_1^k}$. Since $m_{v_1^j}$ is matched to $w_{v_2^k}$ and $k\neq j$, by the way we constructed the preferences, that means that $(v^j, v^k)\in E(G)$ and hence $(m_{v_1^k},w_{v_2^j})\in E(G')$. But that contradicts the social stability of $\mu$ because $w_{v_2^j}$ prefers $m_{v_1^k}$ to $m_{v_2^i}$ (since $k<i$), $m_{v_1^k}$ prefers $w_{v_2^j}$ to $w_{v_1^k}$ (since $w_{v_1^k}$ is his last choice). Hence, it can't be the case that $m_{v_1^j}$ is matched under $\mu$ to someone he prefers to $w_{v_2^i}$.

So it needs to be that $(m_{v_1^j}, w_{v_2^i})\in \mu$ \end{proof}

We can now take all dyads of pairs of the form $\{(m_{v_1^i}, w_{v_2^j}), (m_{v_1^j}, w_{v_2^i})\} $ for $i\neq j$ that belong to $\mu$ and replace them with $\{(m_{v_1^i}, w_{v_2^i}),(m_{v_1^j}, w_{v_2^j})\}$. This will not break social stability of $\mu$ nor decrease it's size. We take now $S=\{v^i \mid m_{v_2^i}\in \mu\}$. Since all $n$ of $m_{v_1^i}$ are matched in $\mu$ there are $k$ couples that contain some $m_{v_2^i}$ and hence $|S|=|\mu|-n=k$. $S$ is an independent set because $m_{v_2^i}$ can only be matched to $w_{v_2^i}$ and hence $m_{v_1^i}$ is matched to $w_{v_1^i}$. The latter pairs are only stable if for every $v^j\in N_G(v^i)$ it holds that $m_{v_1^j}$ is matched to $w_{v_2^j}$ and hence $m_{v_2^j}\notin \mu$ which means that $v^j\notin S$ and hence $S$ is indeed an Independent Set. 
\end{proof}

Lemmas \ref{istosocstable} and \ref{socstabletois} combined give us the following theorem.

\begin{theorem}\label{socstablenphard}
The \textsc{Soc-Stable} problem is {\sf NP}-hard.
\end{theorem}

\par Austrin et al. in \cite{austrin2009inapproximability} show that assuming the unique games conjecture Independent Set is hard to approximate within a factor of $O\left(\frac{d}{\log ^2d}\right)$ for Independent sets of size $k=\left(\frac{1}{2}-\Theta\left(\frac{\log(\log d)}{\log d}\right)\right)n$, where $n$ is the size of the vertex set and $d$ is an upper bound on the degree. Combing that result and Lemmas \ref{istosocstable} and \ref{socstabletois} we get the following corollary.

\begin{corollary}
Under the Unique Game Conjecture the \textsc{Soc-Stable} problem cannot be approximated within a $\frac{3}{2}-\epsilon$, for any constant $\epsilon$.
\end{corollary}
\begin{proof}
The proof follows from the relationship between {\sc Soc-Stable} and Independent Sets shown by Lemmas \ref{istosocstable} and \ref{socstabletois} and the result from \cite{austrin2009inapproximability}. It's exactly the same as the proof of Corollary 4 of \cite{hoefer2012locally}. We set $d=\delta n$ for some constant $\delta>0$ and we get that finding the maximum socially stable matching is hard to approximate within a factor of
\begin{equation}
\frac{n+\left(\frac{1}{2}-\Theta\left(\frac{\log(\log n)}{\log n}\right)\right)n}{n+\left(\frac{1}{2}-\Theta\left(\frac{\log(\log n)}{\log n}\right)\right)n\cdot O\left(\frac{\log n}{n}\right)}\geq \frac{3}{2}-\epsilon
\end{equation}

for sufficiently large $n$.
\end{proof}

\section{Conclusions and Future Work}
In this paper we introduced a generalized model of stable matchings by augmenting the input of a standard matching problem with a graph $G$, which represents the communication between agents in the two sides of the market.
Driven by the intuition the only agents that have some kind of direct communication with each other can coordinate a possible deviation, we defined the notion of {\it social blocking pair} to be a blocking pair that is also connected with an edge in $G$. That gave us an extra degree of flexibility when trying to find maximum {\it socially stable matching}, i.e. matchings that don't have any social blocking pairs, because now, in contrast to the classical stable matching problem, only a subset of the agents can be potentially blocking. We defined the problem {\sc Soc-Stable} to be the problem of finding a maximum size socially stable matching and showed that, assuming the unique games conjecture, this problem is hard to approximate within a ratio of $\frac{3}{2}-\epsilon$, for any constant $\epsilon$. Finally, we accompanied the hardness result with a matching lower bound by showing a $\frac{3}{2}$-approximation algorithm

A lot of interesting future directions arise from this work. First of all it would be interesting to see special cases of the social graph $G$ and/or the preferences of the agents for which the {\sc Soc-Stable} problem is polynomial-time solvable. We believe that when the graph is a tree as well as when there is a master preference list in at least one side of the market, the problem becomes polynomial-time solvable.

One could also look at different optimization problems, other than finding the maximum socially stable matching. For example, if we are allowed to ``block'' at most a certain number of edges of $G$, i.e. remove them from the graph, which edges should we remove in order to obtain as large socially stable matchings as possible?

\bibliography{soc_stable_arxiv}{}
\bibliographystyle{plain}

\end{document}